\documentclass[letterpaper, 10pt, conference]{ieeeconf}
\IEEEoverridecommandlockouts

\usepackage{graphics} 
\usepackage{epsfig} 
\usepackage{amsmath} 
\usepackage{amssymb}  
\usepackage[pdfusetitle]{hyperref}
\usepackage{microtype}
\usepackage{color}

\usepackage{amsthm}

\theoremstyle{plain}
\newtheorem{theorem}{Theorem}
\newtheorem{lemma}{Lemma}
\newtheorem{definition}{Definition}
\newtheorem{remark}{Remark}
\newtheorem{assumption}{Assumption}
\newtheorem{proposition}{Proposition}
\newtheorem{corollary}{Corollary}

\usepackage{pgffor}
\usepackage{pgf}

\foreach \x in {a,b,...,z}
{
  \ifcsname R\x\endcsname
    \PackageWarning{myshortcuts}{Collision: \string\csname R\x\endcsname\space already defined; skipping}%
  \else
    \expandafter\xdef\csname R\x\endcsname{%
      \noexpand\ensuremath{\noexpand\mathbb{R}^{\x}}%
    }%
  \fi

  \ifcsname C\x\endcsname
    \PackageWarning{myshortcuts}{Collision: \string\csname C\x\endcsname\space already defined; skipping}%
  \else
    \expandafter\xdef\csname C\x\endcsname{%
      \noexpand\ensuremath{\noexpand\mathbb{C}^{\x}}%
    }%
  \fi

  \foreach \y in {a,b,...,z}
  {
    \ifcsname R\x t\y\endcsname
      \PackageWarning{myshortcuts}{Collision: \string\csname R\x t\y\endcsname\space already defined; skipping}%
    \else
      \expandafter\xdef\csname R\x t\y\endcsname{%
        \noexpand\ensuremath{\noexpand\mathbb{R}^{\x \times \y}}%
      }%
    \fi
  }
}

\foreach \x in {A,B,...,Z,a,b,...,z}
{
  \ifcsname d\x\endcsname
    \PackageWarning{myshortcuts}{Collision: d\x already defined; skipping}%
  \else
    \expandafter\xdef\csname d\x\endcsname{%
      \noexpand\ensuremath{\noexpand\dot{\x}}%
    }%
  \fi

  \ifcsname \x rm\endcsname
    \PackageWarning{myshortcuts}{Collision: \string\csname \x rm\endcsname\space already defined; skipping}%
  \else
    \expandafter\xdef\csname \x rm\endcsname{%
      \noexpand\ensuremath{\noexpand\mathrm{\x}}%
    }%
  \fi

  \ifcsname \x hat\endcsname
    \PackageWarning{myshortcuts}{Collision: \string\csname \x hat\endcsname\space already defined; skipping}%
  \else
    \expandafter\xdef\csname \x hat\endcsname{%
      \noexpand\ensuremath{\noexpand\hat{\x}}%
    }%
  \fi

  \ifcsname \x what\endcsname
    \PackageWarning{myshortcuts}{Collision: \string\csname \x what\endcsname\space already defined; skipping}%
  \else
    \expandafter\xdef\csname \x what\endcsname{%
      \noexpand\ensuremath{\noexpand\widehat{\x}}%
    }%
  \fi

  \ifcsname \x tild\endcsname
    \PackageWarning{myshortcuts}{Collision: \string\csname \x tild\endcsname\space already defined; skipping}%
  \else
    \expandafter\xdef\csname \x tild\endcsname{%
      \noexpand\ensuremath{\noexpand\tilde{\x}}%
    }%
  \fi

  \ifcsname \x wtild\endcsname
    \PackageWarning{myshortcuts}{Collision: \string\csname \x wtild\endcsname\space already defined; skipping}%
  \else
    \expandafter\xdef\csname \x wtild\endcsname{%
      \noexpand\ensuremath{\noexpand\widetilde{\x}}%
    }%
  \fi

  \ifcsname \x b\endcsname
    \PackageWarning{myshortcuts}{Collision: \string\csname \x b\endcsname\space already defined; skipping}%
  \else
    \expandafter\xdef\csname \x b\endcsname{%
      \noexpand\ensuremath{\noexpand\mathbf{\x}}%
    }%
  \fi

  \ifcsname \x frak\endcsname
  \PackageWarning{myshortcuts}{Collision: \string\csname \x frak\endcsname\space already defined; skipping}%
  \else
    \expandafter\xdef\csname \x frak\endcsname{%
      \noexpand\ensuremath{\noexpand\mathfrak{\x}}%
    }%
  \fi

  \ifcsname \x bh\endcsname
    \PackageWarning{myshortcuts}{Collision: \string\csname \x bh\endcsname\space already defined; skipping}%
  \else
    \expandafter\xdef\csname \x bh\endcsname{%
      \noexpand\ensuremath{\noexpand\boldhat{\x}}%
    }%
  \fi

  \ifcsname \x ol\endcsname
    \PackageWarning{myshortcuts}{Collision: \string\csname \x ol\endcsname\space already defined; skipping}%
  \else
    \expandafter\xdef\csname \x ol\endcsname{%
      \noexpand\ensuremath{\noexpand\overline{\x}}%
    }%
  \fi

  \ifcsname \x bar\endcsname
    \PackageWarning{myshortcuts}{Collision: \string\csname \x bar\endcsname\space already defined; skipping}%
  \else
    \expandafter\xdef\csname \x bar\endcsname{%
      \noexpand\ensuremath{\noexpand\bar{\x}}%
    }%
  \fi

  \ifcsname \x ul\endcsname
    \PackageWarning{myshortcuts}{Collision: \string\csname \x ul\endcsname\space already defined; skipping}%
  \else
    \expandafter\xdef\csname \x ul\endcsname{%
      \noexpand\ensuremath{\noexpand\underline{\x}}%
    }%
  \fi

  \ifcsname \x p\endcsname
    \PackageWarning{myshortcuts}{Collision: \string\csname \x p\endcsname\space already defined; skipping}%
  \else
    \expandafter\xdef\csname \x p\endcsname{%
      \noexpand\ensuremath{\x^{+}}%
    }%
  \fi

  \ifcsname \x sf\endcsname
    \PackageWarning{myshortcuts}{Collision: \string\csname \x sf\endcsname\space already defined; skipping}%
  \else
    \expandafter\xdef\csname \x sf\endcsname{%
      \noexpand\ensuremath{\noexpand\mathsf{\x}}%
    }%
  \fi

  \ifcsname \x bwhat\endcsname
    \PackageWarning{myshortcuts}{Collision: \string\csname \x bwhat\endcsname\space already defined; skipping}%
  \else
    \expandafter\xdef\csname \x bwhat\endcsname{%
      \noexpand\ensuremath{\noexpand\widehat{\noexpand\mathbf{\x}}}%
    }%
  \fi

  \ifcsname \x bol\endcsname
    \PackageWarning{myshortcuts}{Collision: \string\csname \x bol\endcsname\space already defined; skipping}%
  \else
    \expandafter\xdef\csname \x bol\endcsname{%
      \noexpand\ensuremath{\noexpand\overline{\noexpand\mathbf{\x}}}%
    }%
  \fi

  \ifcsname d\x b\endcsname
    \PackageWarning{myshortcuts}{Collision: d\x b already defined; skipping}%
  \else
    \expandafter\xdef\csname d\x b\endcsname{%
      \noexpand\ensuremath{\noexpand\dot{\noexpand\mathbf{\x}}}%
    }%
  \fi

  \ifcsname d\x hat\endcsname
    \PackageWarning{myshortcuts}{Collision: d\x hat already defined; skipping}%
  \else
    \expandafter\xdef\csname d\x hat\endcsname{%
      \noexpand\ensuremath{\noexpand\dot{\noexpand\hat{\x}}}%
    }%
  \fi

  \ifcsname d\x tild\endcsname
    \PackageWarning{myshortcuts}{Collision: d\x tild already defined; skipping}%
  \else
    \expandafter\xdef\csname d\x tild\endcsname{%
      \noexpand\ensuremath{\noexpand\dot{\noexpand\tilde{\x}}}%
    }%
  \fi

  \ifcsname \x hatp\endcsname
    \PackageWarning{myshortcuts}{Collision: \string\csname \x hatp\endcsname\space already defined; skipping}%
  \else
    \expandafter\xdef\csname \x hatp\endcsname{%
      \noexpand\ensuremath{\noexpand\hat{\x}^{+}}%
    }%
  \fi

  \ifcsname \x tildp\endcsname
    \PackageWarning{myshortcuts}{Collision: \string\csname \x tildp\endcsname\space already defined; skipping}%
  \else
    \expandafter\xdef\csname \x tildp\endcsname{%
      \noexpand\ensuremath{\noexpand\tilde{\x}^{+}}%
    }%
  \fi
}

\foreach \x in {A,B,...,Z}
{
  \ifcsname \x cal\endcsname
    \PackageWarning{myshortcuts}{Collision: \expandafter\string\csname \x cal\endcsname\space already defined; skipping}%
  \else
    \expandafter\xdef\csname \x cal\endcsname{%
      \noexpand\ensuremath{\noexpand\mathcal{\x}}%
    }%
  \fi

  \ifcsname \x caltild\endcsname
    \PackageWarning{myshortcuts}{Collision: \expandafter\string\csname \x caltild\endcsname\space already defined; skipping}%
  \else
    \expandafter\xdef\csname \x caltild\endcsname{%
      \noexpand\ensuremath{\noexpand\tilde{\noexpand\mathcal{\x}}}%
    }%
  \fi

  \ifcsname \x calwtild\endcsname
    \PackageWarning{myshortcuts}{Collision: \expandafter\string\csname \x calwtild\endcsname\space already defined; skipping}%
  \else
    \expandafter\xdef\csname \x calwtild\endcsname{%
      \noexpand\ensuremath{\noexpand\widetilde{\noexpand\mathcal{\x}}}%
    }%
  \fi

  \ifcsname \x calhat\endcsname
    \PackageWarning{myshortcuts}{Collision: \expandafter\string\csname \x calhat\endcsname\space already defined; skipping}%
  \else
    \expandafter\xdef\csname \x calhat\endcsname{%
      \noexpand\ensuremath{\noexpand\hat{\noexpand\mathcal{\x}}}%
    }%
  \fi

  \ifcsname \x calwhat\endcsname
    \PackageWarning{myshortcuts}{Collision: \expandafter\string\csname \x calwhat\endcsname\space already defined; skipping}%
  \else
    \expandafter\xdef\csname \x calwhat\endcsname{%
      \noexpand\ensuremath{\noexpand\widehat{\noexpand\mathcal{\x}}}%
    }%
  \fi

  \expandafter\xdef\csname \x\endcsname{%
    \noexpand\ensuremath{\noexpand\mathbb{\x}}%
  }%
}

\foreach \g in {alpha,beta,gamma,delta,epsilon,varepsilon,zeta,eta,theta,vartheta,iota,kappa,lambda,mu,nu,xi,pi,varpi,rho,varrho,sigma,varsigma,tau,upsilon,phi,varphi,chi,psi,omega,Gamma,Delta,Theta,Lambda,Xi,Pi,Sigma,Upsilon,Phi,Psi,Omega}
{
  \ifcsname \g hat\endcsname
    \PackageWarning{myshortcuts}{Collision: \string\csname \g hat\endcsname\space already defined; skipping}%
  \else
    \expandafter\xdef\csname \g hat\endcsname{%
      \noexpand\ensuremath{\noexpand\hat{\csname\g\endcsname}}%
    }%
  \fi

  \ifcsname \g what\endcsname
    \PackageWarning{myshortcuts}{Collision: \string\csname \g what\endcsname\space already defined; skipping}%
  \else
    \expandafter\xdef\csname \g what\endcsname{%
      \noexpand\ensuremath{\noexpand\widehat{\csname\g\endcsname}}%
    }%
  \fi

  \ifcsname \g tild\endcsname
    \PackageWarning{myshortcuts}{Collision: \string\csname \g tild\endcsname\space already defined; skipping}%
  \else
    \expandafter\xdef\csname \g tild\endcsname{%
      \noexpand\ensuremath{\noexpand\tilde{\csname\g\endcsname}}%
    }%
  \fi

  \ifcsname \g wtild\endcsname
    \PackageWarning{myshortcuts}{Collision: \string\csname \g tild\endcsname\space already defined; skipping}%
  \else
    \expandafter\xdef\csname \g wtild\endcsname{%
      \noexpand\ensuremath{\noexpand\widetilde{\csname\g\endcsname}}%
    }%
  \fi

  \ifcsname \g ol\endcsname
    \PackageWarning{myshortcuts}{Collision: \string\csname \g ol\endcsname\space already defined; skipping}%
  \else
    \expandafter\xdef\csname \g ol\endcsname{%
      \noexpand\ensuremath{\noexpand\overline{\csname\g\endcsname}}%
    }%
  \fi

  \ifcsname \g bar\endcsname
    \PackageWarning{myshortcuts}{Collision: \string\csname \g bar\endcsname\space already defined; skipping}%
  \else
    \expandafter\xdef\csname \g bar\endcsname{%
      \noexpand\ensuremath{\noexpand\bar{\csname\g\endcsname}}%
    }%
  \fi

  \ifcsname \g ul\endcsname
    \PackageWarning{myshortcuts}{Collision: \string\csname \g ul\endcsname\space already defined; skipping}%
  \else
    \expandafter\xdef\csname \g ul\endcsname{%
      \noexpand\ensuremath{\noexpand\underline{\csname\g\endcsname}}%
    }%
  \fi

  \ifcsname d\g\endcsname
    \PackageWarning{myshortcuts}{Collision: d\g already defined; skipping}%
  \else
    \expandafter\xdef\csname d\g\endcsname{%
      \noexpand\ensuremath{\noexpand\dot{\csname\g\endcsname}}%
    }%
  \fi

  \ifcsname \g p\endcsname
    \PackageWarning{myshortcuts}{Collision: \string\csname \g p\endcsname\space already defined; skipping}%
  \else
    \expandafter\xdef\csname \g p\endcsname{%
      \noexpand\ensuremath{\csname\g\endcsname^{+}}%
    }%
  \fi

  \ifcsname d\g hat\endcsname
    \PackageWarning{myshortcuts}{Collision: d\g hat already defined; skipping}%
  \else
    \expandafter\xdef\csname d\g hat\endcsname{%
      \noexpand\ensuremath{\noexpand\dot{\noexpand\hat{\csname\g\endcsname}}}%
    }%
  \fi

  \ifcsname d\g tild\endcsname
    \PackageWarning{myshortcuts}{Collision: d\g tild already defined; skipping}%
  \else
    \expandafter\xdef\csname d\g tild\endcsname{%
      \noexpand\ensuremath{\noexpand\dot{\noexpand\tilde{\csname\g\endcsname}}}%
    }%
  \fi

  \ifcsname \g hatp\endcsname
    \PackageWarning{myshortcuts}{Collision: \string\csname \g hatp\endcsname\space already defined; skipping}%
  \else
    \expandafter\xdef\csname \g hatp\endcsname{%
      \noexpand\ensuremath{\noexpand\hat{\csname\g\endcsname}^{+}}%
    }%
  \fi

  \ifcsname \g tildp\endcsname
    \PackageWarning{myshortcuts}{Collision: \string\csname \g tildp\endcsname\space already defined; skipping}%
  \else
    \expandafter\xdef\csname \g tildp\endcsname{%
      \noexpand\ensuremath{\noexpand\tilde{\csname\g\endcsname}^{+}}%
    }%
  \fi
}

\newcommand{\Rgzero}{\mathbb{R}_{\geq 0}}

\newcommand{\Zgzero}{\mathbb{Z}_{\geq 0}}

\newcommand{\Ball}{\mathbb{B}}

\newcommand{\mtset}{\emptyset}

\newcommand{\SH}{\mathcal{S}_\mathcal{H}}

\newcommand{\eps}{\varepsilon}

\DeclareMathOperator*{\cch}{\overline{\mathrm{co}}}

\DeclareMathOperator*{\dom}{dom}

\DeclareMathOperator*{\graph}{graph}
\DeclareMathOperator*{\range}{rge}

\DeclareMathOperator*{\olim}{omega}

\newcommand{\Pspace}{(\Omega,\mathcal{F},\mathbb{P})}
\newcommand{\Exc}[2]{\mathbb{E}\left[#1 \mid #2 \right]}

\newcommand{\tto}{\rightrightarrows}

\newcommand{\inner}[1]{\langle#1\rangle}
\newcommand{\ol}[1]{\overline{#1}}

\newcommand{\jmol}{\ol{\jmath}}

\newcommand{\proj}{\mathrm{proj}}

\let\inf\relax \DeclareMathOperator*\inf{\vphantom{p}inf}

\newcommand{\Rnf}{\mathbb{R}^{n_f}}
\newcommand{\Rns}{\mathbb{R}^{n_s}}

\newcommand{\jk}{\overline{\jmath}_k}
\newcommand{\ji}{\overline{\jmath}_i}
\newcommand{\kj}{\overline{k}_j}
\newcommand{\HTT}{\mathcal{H}_{\mathrm{TT}}^{1}}
\newcommand{\HTTe}{\mathcal{H}_{\mathrm{TT}}^{\eps}}

\newcommand{\HBL}{\mathcal{H}_{\mathrm{BL}}}
\newcommand{\HBLK}{\mathcal{H}_{\mathrm{BL}}^{K}}
\newcommand{\OmegaBLK}{\Omega_{\HBLK}}

\newcommand{\HRL}{\mathcal{H}_{\mathrm{R}}^{\Lambda}}
\newcommand{\CR}{C_{\mathrm{R}}^{\Lambda}}
\newcommand{\FR}{F_{\mathrm{R}}^{\Lambda}}
\newcommand{\DR}{D_{\mathrm{R}}^{\Lambda}}
\newcommand{\GR}{G_{\mathrm{R}}^{\Lambda}}
\newcommand{\LRC}{\Lambda_{\mathrm{R}}^{C}}
\newcommand{\LRD}{\Lambda_{\mathrm{R}}^{D}}

\newcommand{\HHB}{\mathcal{H}_{\mathrm{HB}}}
\newcommand{\CH}{C_{\mathrm{HB}}}
\newcommand{\fH}{F_{\mathrm{HB}}}

\renewcommand{\DH}{D_{\mathrm{HB}}}
\newcommand{\gH}{G_{\mathrm{HB}}}

\newcommand{\HTHBe}{\mathcal{H}_{\mathrm{2HB}}^{\varepsilon}}
\newcommand{\HTHB}{\mathcal{H}_{\mathrm{2HB}}^{1}}
\newcommand{\FTHBe}{F_{\mathrm{2HB}}^{\varepsilon}}
\newcommand{\FTHB}{F_{\mathrm{2HB}}^{1}}
\newcommand{\GTHB}{G_{\mathrm{2HB}}}

\title{\LARGE \bf Two-Timescale Asymptotic Simulations of Hybrid Inclusions with Applications to Stochastic Hybrid Optimization}

\author{Max F. Crisafulli and Andrew R. Teel
\thanks{Max F. Crisafulli and Andrew R. Teel are with the ECE Department, University of California, Santa Barbara, CA 93106-9560, USA. (maxcrisafulli@ucsb.edu, teel@ucsb.edu). Research supported in part by ARO grant number W911NF-26-1-0001 and AFOSR grant number FA9550-25-1-0186.}}

\begin{document}
\maketitle
\thispagestyle{empty}
\pagestyle{empty}

\begin{abstract}
   Convergence properties of model-free two-timescale asymptotic simulations of singularly perturbed hybrid inclusions are developed.
A hybrid inclusion combines constrained differential and difference inclusions to capture continuous (flow) and discrete (jump) dynamics, respectively. Sufficient conditions are established under which sequences of iterates and step sizes constitute a two-timescale asymptotic simulation of such a system, with limiting behavior characterized via weakly invariant and internally chain-transitive sets of an associated boundary layer and reduced system.
To illustrate the applicability of these results, conditions are given under which a two-timescale stochastic approximation of a hybrid optimization algorithm asymptotically recovers the behavior of its deterministic counterpart.
\end{abstract}
\renewcommand{\proj}{\mathrm{P}}

\section{Introduction}
This work extends the notion of asymptotic simulations developed for hybrid inclusions in \cite{GoebelChainTransitivity2026} and \cite{AsymptoticSim_CDC26_Short} to the two-timescale setting. This development was itself motivated by the notion of asymptotic pseudotrajectories developed in \cite{BenaimAPT1996} and \cite{BenaimSA2005}, where convergence properties of stochastic approximations of differential inclusions were characterized in terms of internally chain-transitive sets of an associated differential inclusion.

In contrast to existing results such as \cite{TeelARC2025}, which relied on Lyapunov-based arguments, or \cite{TeelCDC2024}, which imposed mild graphical conditions on the data of an explicit simulator and underlying hybrid inclusion, the results presented here characterize limiting behavior entirely in terms of the weakly invariant and internally chain transitive sets of the underlying hybrid inclusion being approximated.

The existing literature on two-timescale stochastic approximation, notably \cite{BorkarTTSA1997} and its recent extension \cite{BorkarTTSAGeneral}, has thus far been restricted to the setting of differential inclusions. The present work extends these results to hybrid inclusions, a broader class of dynamical systems. This extension is motivated by potential applications to the stochastic approximation of reset-based or event-triggered optimization algorithms, an example of which, see \cite{LeHHB2021}, is presented in Section~\ref{sec:example}.

A further advantage of the framework developed here is that it does not rely on the specification of an explicit simulator model. Rather, the analysis provides conditions directly on a sequence of iterates and their associated step sizes, under which they constitute an asymptotic simulation of the limiting hybrid inclusion.
Furthermore, a key motivation in the development of these results is that they enable the approximation of hybrid systems whose flow and jump sets depend on states that may not be directly available or may be corrupted by noise.

To illustrate the applicability of the developed results, we present an example in which a deterministic hybrid optimization algorithm, developed in \cite{LeHHB2021}, is implemented with stochastic gradient descent in a two-timescale manner so as to asymptotically recover its deterministic behavior.
For brevity, many of the proofs are omitted and are instead reported in \cite{TTAS_arXiv}.

\emph{Notation:} $\Rgzero$ ($\Zgzero$) denotes the nonnegative real numbers (integers). For $\ell \in \Zgzero$, $\Z_{\geq \ell}$ is the set of integers that are greater than or equal to $\ell$. The closed unit ball centered at the origin is denoted $\Ball \subset \Rn$ and, given $\rho > 0$, $\rho\Ball$ denotes the closed ball of radius $\rho$. Given a point $x \in \Rn$ and a set $\Acal \subset \Rn$, let $|x|_{\Acal} := \inf_{y \in \Acal}|x-y|$. Given two vectors $v,w$, let $(v,w) := (v^\top,w^\top)^\top$, i.e., the stack of $v$ and $w$. We adopt the set-valued terminology and notation of \cite{VariationalAnalysis}. The graph of a set-valued mapping $S : \Rn \tto \Rm$ is defined as $\graph(S) := \left\{(x,y) \in \Rn \times \Rm : y \in S(x)\right\}$.

\section{Hybrid Inclusions}
A hybrid inclusion with state $x \in \Rn$ is represented by
\begin{equation}\label{eq:hsys}
    \begin{aligned}
        x &\in C \quad \dx \in F(x) \\
        x &\in D \quad \xp \in G(x).
    \end{aligned}
\end{equation}
The set $C \subset \Rn$ is called the flow set, the set-valued mapping $F : \Rn \tto \Rn$ is called the flow map, the set $D \subset \Rn$ is called the jump set, and the set-valued mapping $G : \Rn \tto \Rn$ is called the jump map. We refer to \eqref{eq:hsys} as $\Hcal$, with \emph{data} $\Hcal = (C,F,D,G)$.
A hybrid arc $\psi : \dom\psi \to \Rn$ is a \emph{solution} to $\Hcal$ if $\psi(0,0) \in C\cup D$, and
\begin{itemize}
    \item for every $([t_1,t_2],j) \subset \dom\psi$ with $t_1 < t_2$, $\psi(t,j) \in C$ and $\dpsi(t,j) \in F(\psi(t,j))$ for almost all $t \in [t_1,t_2]$;
    \item for every $(t,j),(t,j+1) \in \dom\psi$, $\psi(t,j) \in D$ and $\psi(t,j+1) \in G(\psi(t,j))$.
\end{itemize}
The set of all solutions to $\Hcal$ is denoted $\SH$, and each such solution $\psi \in \SH$ is called an $\SH$-arc. The set of solutions $\psi\in\SH$ with $\psi(0,0) \in X \subset \Rn$ is denoted $\SH(X)$. A solution $\psi \in \SH$ is \emph{complete} if $\dom\psi$ is unbounded, and \emph{maximal} if it cannot be extended. To ensure the regularity of solutions to $\Hcal$, we impose the following assumption.

\begin{assumption}[{\cite[Assumption~6.5]{HDS_Book}}]
\label{ass:basic_conditions}
    The hybrid system $\Hcal = (C,F,D,G)$ satisfies the \emph{Hybrid Basic Conditions} if:
    \begin{itemize}
        \item The sets $C,D \subset \Rn$ are closed.
        \item The mappings $F,G : \Rn \tto \Rn$ are outer semicontinuous and locally bounded.
        \item For each $x \in C$, $F(x)$ is nonempty and convex.
        \item For each $x \in D$, $G(x)$ is nonempty.
    \end{itemize}
\end{assumption}

\subsection{Weak Invariance and Internal Chain Transitivity}
The following definitions are used to characterize the asymptotic behavior of solutions to $\Hcal = (C,F,D,G)$, whose data is assumed to satisfy Assumption~\ref{ass:basic_conditions}. Weak invariance of compact sets for a hybrid system $\Hcal$ is recalled here succinctly; see \cite[Def.~6.19]{HDS_Book} for a more detailed treatment.

\begin{definition}
    A compact set $K \subset \Rn$ is said to be \emph{weakly $\SH$-invariant} if, for all $x \in K$ and $T > 0$, there exists a complete $\psi \in \SH$ and $(s,i)\in\dom\psi$, with $s + i \geq T$, such that $\psi(t,j) \in K$ for all $(t,j) \in \dom\psi$ and $\psi(s,i) = x$.
\end{definition}

The notion of $(\tau,\eps)$-$\SH$-chains and chain recurrence for hybrid systems was introduced in \cite{GoebelDirectProof2023}, and later generalized to chain transitivity for collections of hybrid arcs in \cite{{GoebelChainTransitivity2026}}. We recall the basics here.

\begin{definition}
    Given two points $x,y \in \Rn$ and $\tau,\eps > 0$, a \emph{$(\tau,\eps)$-$\SH$-chain} from $x$ to $y$ consists of finite sequences of points $x = x_0,x_1,\ldots,x_{k^*} = y$ in $\Rn$ and solutions $\psi_0,\psi_1,\ldots,\psi_{k^*-1} \in \SH$ such that, for $0 \leq k < k^{*}$,  $\psi_k(0,0) = x_k$ and $x_{k+1} \in \psi_k(t_k,j_k) + \eps\Ball$ for some $(t_k,j_k) \in \dom \psi_k$ with $t_k + j_k \geq \tau$.
    When considering a nonempty and closed set $K \subset \Rn$ and given two points $x,y \in K$ and $\tau,\eps > 0$, an \emph{internal $(\tau,\eps)$-$\SH$-chain from $x$ to $y$} is a $(\tau,\eps)$-$\SH$-chain from $x$ to $y$ such that $\psi_k(t,j) \in K$ for all $(t,j) \in \dom \psi_k$ with $0 \leq t + j \leq t_k + j_k$ for $0 \leq k < k^*$.
\end{definition}

\begin{definition}
    A set $K \subset \Rn$ is \emph{internally $\SH$-chain transitive} if, for every $x,y \in K$ and every $\tau,\eps > 0$, there exists an internal $(\tau,\eps)$-$\SH$-chain from $x$ to $y$.
\end{definition}

\section{Asymptotic Simulations of Hybrid Inclusions}
We recall here relevant content from \cite[Sec.~5]{GoebelChainTransitivity2026}. Although the present work extends the single-timescale results of \cite{GoebelChainTransitivity2026} to the two-timescale setting for singularly perturbed hybrid systems, the single-timescale formulation remains essential; in particular, it is used in the asymptotic simulation analysis of both the boundary layer and reduced systems that arise in the study of the two-timescale system.

\subsection{Hybrid Sequences}
A set $E \subset \Zgzero^2$ is a \emph{compact hybrid sequence domain} if 
\[E = \bigcup_{j=0}^{J-1}\Bigl(\{k_j,\ldots,k_{j+1}\}\times\{j\}\Bigr),\]
where $J \in \Zgzero$ and $0 = k_0 \leq k_1 \leq \cdots \leq k_J$ form a finite sequence of integers. It is a \emph{hybrid sequence domain} if it is the union of a nondecreasing sequence of compact hybrid sequence domains.

A mapping $\phi : \dom\phi \to \Rn$ is a \emph{hybrid sequence} if its domain is a hybrid sequence domain. It is \emph{complete} if its domain is unbounded. It is \emph{complete in the $k$-direction} if the set of all $k$'s defining its domain is unbounded and \emph{complete in the $j$-direction} if the set of all $j$'s defining its domain is unbounded. Given a hybrid sequence $\phi$, let 
\begin{equation}\label{eq:jk_and_kj}
    \begin{aligned}
        \jk &:= \inf\{j \in \Zgzero : (k+1,j) \in \dom \phi\}  \\ 
        \kj &:= \inf\{k \in \Zgzero : (k,j+1) \in \dom \phi\}
    \end{aligned}
\end{equation}
with the convention that $\inf(\mtset) = \infty$.
As noted in \cite{GoebelChainTransitivity2026}, if $\phi$ is complete in the $k$-direction then the value $\jk$ is finite for each $k \in \Zgzero$ and $(k,j),(k+1,j) \in \dom\phi$ if and only if $j = \jk$. Similarly, if $\phi$ is complete in the $j$-direction then the value $\kj$ is finite for every $j \in \Zgzero$ and $(k,j),(k,j+1) \in \dom\phi$ if and only if $k = \kj$.
To facilitate an efficient description of asymptotic simulations, define, given a system $\Hcal = (C,F,D,G)$, the mappings
\[F_C(x) := \begin{cases}F(x), &x \in C \\ \mtset, &x \notin C\end{cases}, \quad G_D(x) := \begin{cases}G(x), &x \in D \\ \mtset, &x \notin D.\end{cases}\]
Note that
\[
    \begin{aligned}
        \graph(F_C) &= \graph(F)\cap(C\times\Rn),  \\
        \graph(G_D) &= \graph(G)\cap(D\times\Rn).
    \end{aligned}
\]

\subsection{Asymptotic Simulations}
An asymptotic simulation of $\Hcal$, denoted by the pair $(\phi,\{h_k\}_{k=1}^\infty)$, includes a hybrid sequence $\phi$ and a sequence of converging, positive step sizes $\{h_k\}_{k=1}^\infty$ that play a role in approximating the flows of a hybrid inclusion. An analogous concept, the \emph{asymptotic pseudo-trajectory}, was first developed in \cite{BenaimAPT1996} and later applied to the analysis of stochastic approximations of differential inclusions in \cite{BenaimSA2005}.

\begin{definition} 
\label{def:admissible}
A sequence of step sizes $\{h_k\}_{k=1}^\infty$ is said to be \emph{admissible} if $h_{k} > 0$ for each $k \in \Z_{\geq 1}$
and the sequence converges to zero but is not summable. Given an admissible sequence of step sizes, define
\[\tau_k := \sum_{i=0}^{k-1}h_{i+1} \quad \forall k \in \Zgzero,\]
and $m(t) := \max\{k \in \Zgzero : \tau_k \leq t\}$ for all $t \in \Rgzero$.
\end{definition}
The following definition is taken from \cite{GoebelChainTransitivity2026}, and introduces the notion of a (single-timescale) asymptotic simulation.
\begin{definition}\label{def:asymptotic_simulation}
    $(\phi,\{h_k\}_{k=1}^\infty)$ is an \emph{asymptotic simulation} of $\Hcal$ if $\phi$ is a bounded, complete hybrid sequence, $\{h_k\}_{k=1}^\infty$ is an admissible sequence of step sizes, and the following properties hold:
    \begin{enumerate}
        \item If $\phi$ is complete in the $k$-direction then there exists a bounded sequence $\{f_k\}_{k=0}^\infty$ of vectors in $\Rn$ such that 
        \begin{equation}\label{eq:asymptotic_sim_F_containment}
            \limsup_{k \to \infty}\left(\phi(k,\jk),f_k\right) \subset \graph(F_C),
        \end{equation}
        and, with the definition 
        \[\fwhat_{k+1}:=\frac{\phi(k+1,\jk) - \phi(k,\jk)}{h_{k+1}} \quad \forall k \in \Zgzero,\]
        the following limit holds for each $T > 0$: 
        \begin{equation}\label{eq:asymptotic_sim_closeness}
            \lim_{n \to \infty}\sup_{n+1 \leq k \leq m(\tau_n+T)}\left|\sum_{i=n}^{k-1}h_{i+1}(\fwhat_{i+1}-f_i)\right| = 0.
        \end{equation}
        \item If $\phi$ is complete in the $j$-direction then 
        \begin{equation}\label{eq:asymptotic_sim_G_containment}
            \limsup_{j \to \infty}\left(\phi(\kj,j),\phi(\kj,j+1)\right) \subset \graph(G_D).
        \end{equation}
    \end{enumerate}
\end{definition}
The $\limsup$ appearing in \eqref{eq:asymptotic_sim_F_containment} and \eqref{eq:asymptotic_sim_G_containment} is the \emph{outer limit}, i.e., the set of all accumulation points of the considered sequence; see \cite[Ch.~4.A]{VariationalAnalysis} for a rigorous definition.

It can be verified that a hybrid inclusion implementing a forward Euler approximation of the flows of $\Hcal$ constitutes an asymptotic simulation of $\Hcal$, as illustrated for the two-timescale setting in Remark~\ref{remark:TTmodel}.
The $\omega$-limit set of a hybrid sequence $\phi$ is the set $\omega(\phi)$ defined as
\[
    \lim_{i \to \infty}\bigl\{z \in \Rn : z = \phi(k,j),\; (k,j) \in \dom \phi,\; k+j \geq i\bigr\}.
\]
This set is closed, and if $\phi$ is complete and bounded, it is nonempty, bounded, and thus compact, with the property that $\phi$ converges to $\omega(\phi)$. With the assumption that $(\phi,\{h_k\}_{k=1}^\infty)$ is indeed an asymptotic simulation of $\Hcal$, the following statement can be made about its $\omega$-limit set. 

\begin{proposition}\label{prop:aymptotic_sim_omega}
    Let the system $\Hcal$ satisfy Assumption~\ref{ass:basic_conditions} and $(\phi,\{h_k\}_{k=1}^\infty)$ be an asymptotic simulation of $\Hcal$. Then, $\omega(\phi)$ is a nonempty and compact set that is both weakly $\SH$-invariant and internally $\SH$-chain transitive.
\end{proposition}

This result was developed in \cite[Thm.~5.1]{GoebelChainTransitivity2026} and can be considered a hybrid extension of \cite[Thm.~4.3]{BenaimSA2005}, which was established in the setting of differential inclusions.

\section{Two-Timescale Asymptotic Simulations}
In the following sections, we distinguish between the iterates associated with the slow and fast dynamics of a coupled asymptotic simulation, denoted by the subscripts $s$ and $f$, respectively. For definitions common to both timescales, we use the variable $r \in \{s,f\}$.

\subsection{Two-Timescale System}
Consider a hybrid inclusion of the form
\begin{equation}\label{eq:two_timescale_sys}
    \begin{aligned}
        (x_s,x_f) &\in C \quad \begin{cases}\dx_s \in F_s(x_s,x_f) \\ \eps\dx_f \in F_f(x_s,x_f)\end{cases} \\
        (x_s,x_f) &\in D \quad (\xp_s,\xp_f) \in G(x_s,x_f),
    \end{aligned}
\end{equation}
where $\eps > 0$ is a small parameter, $x_s \in \Rns$ is the ``slow'' state, $x_f \in \Rnf$ is the ``fast'' state, and $x := (x_s,x_f) \in \Rn$ denotes the overall state, with $n_s + n_f = n$. We denote the overall flow map as
\[F^{\eps}(x) := F_s(x_s,x_f)\times \eps^{-1}F_f(x_s,x_f),\]
and refer to \eqref{eq:two_timescale_sys} as a \emph{two-timescale system}, or $\HTTe = (C,F^{\eps},D,G)$, with the data assumed to satisfy Assumption~\ref{ass:basic_conditions}. 
The $\eps$ parameter does not appear in the system that we are ultimately to approximate; rather, we simulate $\HTT = (C,F^{1},D,G)$, with the separation of timescales introduced in the following definition capturing the two-timescale behavior of $\HTTe$ as $\eps \to 0^+$.
In what follows, let $\proj_{r}(X)$, for $r \in \{s,f\}$, denote the projection of a set $X \subset \Rn$ onto $\R^{n_r}$. The following definition generalizes Definition~\ref{def:admissible} to the two-timescale setting.

\begin{definition}\label{def:two_timescale_admissible}
    A sequence of step size vectors $\{H_k\}_{k=1}^\infty$, with
    \[H_k := (h_{s,k}, h_{f,k}) \in \R^2_{>0} \quad \forall k \in \Z_{\geq 1},\]
    is said to be \emph{two-timescale admissible} if, for $r\in\{s,f\}$, the step sizes $\{h_{r,k}\}_{k=1}^\infty$ are themselves admissible and
    \begin{equation}\label{eq:timescale_sep}
        \lim_{k \to \infty}\frac{h_{s,k}}{h_{f,k}} = 0.
    \end{equation}
    Given a two-timescale admissible sequence of step sizes, define, for $r \in \{s,f\}$, the quantities 
    \[\tau_{r,k} := \sum_{i=0}^{k-1}h_{r,i+1} \quad \forall k \in \Zgzero\]
    and $m_r(t) := \max\{k \in \Zgzero : \tau_{r,k} \leq t\}$ for all $t \in \Rgzero$.
\end{definition}
The quantity $\tau_{r,k}$ represents the accumulated time on the $r$-timescale at the $k$-th step, while $m_r(t)$ returns the largest index $k$ for which the accumulated time on the $r$-timescale does not exceed $t$. 
The assumption that $\lim_{k \to \infty}h_{s,k}/h_{f,k} = 0$, first formalized in the two-timescale stochastic approximation literature in \cite{BorkarTTSA1997}, and which may be seen as a discrete-time analogue of the $\eps \to 0^+$ singular-perturbation formulation, enforces a separation of timescales for the simulated flow dynamics. 
In particular, for sufficiently large $k$, this implies $\tau_{s,k} \ll \tau_{f,k}$, meaning the fast iterates evolve on a much faster timescale than the slow iterates. Consequently, the fast iterates see the slow ones as quasi-static, while the slow iterates see the steady-state behavior of the fast iterates.
For each $(r,n,T) \in \{s,f\}\times\Zgzero\times\R_{>0}$, let $\Ical_{r,n,T}$ denote the index set
\[\Ical_{r,n,T} := \{k \in \Zgzero : n+1 \leq k \leq m_r(\tau_{r,n} + T)\}.\]

The notion of a two-timescale asymptotic simulation of the hybrid inclusion $\HTT$, given next, generalizes Definition~\ref{def:asymptotic_simulation}.
\begin{definition}\label{def:two_timescale_sim}
    $(\Phi,\{H_{k}\}_{k=1}^\infty)$ is a \emph{two-timescale asymptotic simulation} of $\HTT$ if $\Phi$ is a bounded, complete hybrid sequence, $\{H_k\}_{k=1}^\infty$ is a two-timescale admissible sequence of step sizes, and, with the definitions
    \[\Phi(k,j) := \begin{bmatrix}
        \phi_{s}(k,j) \\ \phi_{f}(k,j)
    \end{bmatrix}, \quad f_{k} := \begin{bmatrix}
        f_{s,k} \\ f_{f,k}
    \end{bmatrix},\]
    the following properties hold:
    \begin{enumerate}
        \item If $\Phi$ is complete in the $k$-direction, then there exists a bounded sequence $\{f_{k}\}_{k=0}^\infty$ of vectors in $\Rn$ such that
        \begin{equation}\label{eq:two_timescale_F_containment}
            \limsup_{k \to \infty}\left(\Phi(k,\jk),f_{k}\right) \subset \graph(F^{1}_C),
        \end{equation}
        and        
        with $\fwhat_{r,k+1}$, for $r \in \{s,f\}$, defined as
        \[\fwhat_{r,k+1} := \frac{\phi_{r}(k+1,\jk)-\phi_{r}(k,\jk)}{h_{r,k+1}} \quad \forall k \in \Zgzero,\]
        the following limits hold for each $T > 0$:
        \begin{subequations}\label{eq:two_timescale_closeness}
            \begin{align}
                \lim_{n \to \infty}\sup_{k \in \Ical_{s,n,T}}&\left|\sum_{i=n}^{k-1}h_{s,i+1}(\fwhat_{s,i+1}-\,f_{s,i})\right| = 0 \label{eq:slow_closeness}\\
                \lim_{n \to \infty}\sup_{k \in \Ical_{f,n,T}}&\left|\sum_{i=n}^{k-1}h_{f,i+1}(\fwhat_{f,i+1}-f_{f,i})\right| = 0, \label{eq:fast_closeness}
            \end{align}
        \end{subequations}
        with $m_r(t)$ and $\tau_{r,n}$ as defined in Definition~\ref{def:two_timescale_admissible}.
        \item If $\Phi$ is complete in the $j$-direction then
        \begin{equation}\label{eq:two_timescale_G_containment}
            \begin{aligned}
                \limsup_{j \to \infty}\left(\Phi(\kj,j), \Phi(\kj,j+1)\right) \subset \graph(G_D).
            \end{aligned}
        \end{equation}
    \end{enumerate}
\end{definition}

\begin{remark}\label{remark:TTmodel}
    Consider an explicit simulator implementing a forward Euler approximation of the flows of $\HTT$, i.e.,
    \begin{equation}\label{eq:euler_TTsys}
        \begin{aligned}
            x &\in C \quad \begin{cases}\xp_s - x_s \in \hp_s\Fwhat_s(x,\hp_s) \\ 
            \xp_f - x_f \in \hp_f \Fwhat_f(x,\hp_f)\end{cases} \\
            x &\in D \quad \xp \in G(x).
        \end{aligned}
    \end{equation}
    A \emph{solution} of \eqref{eq:euler_TTsys} is a hybrid sequence $\Phi : \dom\Phi \to \Rns \times \Rnf$ such that, for $r \in \{s,f\}$,
    \begin{itemize}
        \item $(k,j),(k+1,j) \in \dom\Phi$ implies $\Phi(k,j) \in C$ and 
        \[\phi_{r}(k+1,j) - \phi_{r}(k,j) \in h_{r,k+1}\Fwhat_r(\Phi(k,j),h_{r,k+1}).\]
        \item $(k,j),(k,j+1) \in \dom \Phi$ implies $\Phi(k,j) \in D$ and \[\Phi(k,j+1) \in G(\Phi(k,j)).\]
    \end{itemize}
    Following \cite[Eq.~(12)]{GoebelSmooth2025}, assume that there exists a continuous, nondecreasing function $\gamma : \Rgzero \to \Rgzero$ such that, for all $(x,k,r) \in C \times \Zgzero \times \{s,f\}$, it holds that
    \[\Fwhat_r(x,h_{r,k+1}) \subset F_r(x) + h_{r,k+1}\gamma(|x|)\Ball.\]
    Then, it can be verified that when $\{H_k\}_{k=1}^\infty$ is two-timescale admissible and $\Phi$ is bounded and complete, the pair $(\Phi,\{H_k\}_{k=1}^\infty)$ is a two-timescale asymptotic simulation of $\HTT$. In particular, taking $f_{r,k} := \fwhat_{r,k+1}$ for all $k \in \Zgzero$ satisfies \eqref{eq:two_timescale_F_containment} and \eqref{eq:two_timescale_closeness}, with \eqref{eq:two_timescale_G_containment} holding trivially.
\end{remark}

\subsection{Boundary Layer System}
Consider the hybrid inclusion
\begin{equation}\label{eq:bl_sys}
    \begin{aligned}
        x &\in C \quad \dx \in \Fwtild(x) \\
        x &\in D \quad \xp \in G(x)
    \end{aligned}
\end{equation}
where 
\[\Fwtild(x) := \{0\} \times F_f(x).\]
We refer to \eqref{eq:bl_sys} as the \emph{boundary layer system}, or $\HBL= (C,\Fwtild,D,G)$. Under Assumption~\ref{ass:basic_conditions} on $\HTT$, the system $\HBL$ satisfies the same regularity conditions. The following theorem constitutes the first novel contribution of this paper.
\begin{theorem}\label{theorem:HBL_sim}
    Let $(\Phi,\{H_k\}_{k=1}^\infty)$ be a two-timescale asymptotic simulation of $\HTT$. Then, $(\Phi,\{h_{f,k}\}_{k=1}^\infty)$ is an asymptotic simulation of $\HBL$.
\end{theorem}

Let $(\Phi,\{H_k\}_{k=1}^\infty)$ be a two-timescale asymptotic simulation of $\HTT$; in particular, $\Phi$ is bounded and thus $\omega(\Phi)$ is nonempty and compact. Let $\Lambda : \Rns \tto \Rnf$ be a set-valued mapping whose graph is compact and contains $\omega(\Phi)$.

\begin{remark} If $\Phi$ is complete in the $k$-direction, then
    \begin{equation}\label{eq:limsup_k_A}
        \limsup_{k \to \infty}(\phi_{s}(k,\jk),\phi_{f}(k,\jk)) \subset \omega(\Phi) \subset \graph(\Lambda).
    \end{equation}
    Similarly, if $\Phi$ is complete in the $j$-direction, then
    \begin{equation}\label{eq:limsup_j_A}
        \limsup_{j \to \infty}(\phi_{s}(\kj,j),\phi_{f}(\kj,j)) \subset \omega(\Phi) \subset \graph(\Lambda).
    \end{equation}
\end{remark}
Since $\graph(\Lambda)$ is compact (and hence closed), outer-semicontinuity of $\Lambda$ follows from \cite[Thm.~5.7(a)]{VariationalAnalysis}. Moreover, the compactness (and hence boundedness) of $\range(\Lambda)$ implies local boundedness of $\Lambda$ by \cite[Prop.~5.15]{VariationalAnalysis}. Therefore, the mapping $\Lambda$ is outer-semicontinuous and locally bounded. 

Let $\HBLK$ be the hybrid inclusion $\HBL$ restricted to $K\subset\Rn$, i.e., $\HBLK = (C\cap K,\Fwtild,D^K,G^K)$, where $D^K := \{x \in D\cap K : G(x)\cap K \neq \mtset\}$ and $G^K(x) := G(x)\cap K$ for every $x \in \Rn$. The system $\HBLK$ satisfies Assumption~\ref{ass:basic_conditions} if $K$ is compact. We now adapt \cite[Def.~6.23]{HDS_Book}.

\begin{definition}
    Define, for a set $X \subset \Rn$, $\OmegaBLK(X)$ as the set of all $z \in \Rn$ for which there exists a sequence $\{\psi_i\}_{i=1}^\infty$ of solutions $\psi_i \in \Scal_{\HBLK}(X)$ and a sequence $\{(t_i,j_i)\}_{i=1}^\infty$ of points $(t_i,j_i) \in \dom\psi_i$ such that $\lim_{i\to\infty}t_i+j_i = \infty$ and $\lim_{i \to \infty}\psi_i(t_i,j_i) = z$.
\end{definition}

In settings where $\omega(\Phi)$ cannot be determined solely from the behavior of $\HTT$, the following lemma allows for the construction of
$\Lambda$ using additional information about $\HBL$. 

\begin{lemma}\label{lemma:HBL_omegalim}
    Let $(\Phi,\{h_{f,k}\}_{k=1}^\infty)$ be an asymptotic simulation of $\HBL$ and $K$ be a compact set for which $\Phi(k,j) \in K$ for all $(k,j) \in \dom\Phi$. Then, $\omega(\Phi) \subset \OmegaBLK(K)$.
\end{lemma}

\begin{remark}    
    It should be emphasized that $\HBL$ is not a subset of $\HTT$, and solutions of $\HBL$ need not correspond to solutions of $\HTT$. Rather, the connection arises via the fact that a two-timescale asymptotic simulation of $\HTT$ is also an asymptotic simulation of $\HBL$. Consequently, $\omega(\Phi)$ is additionally constrained by weakly $\Scal_{\HBL}$-invariant and internally $\Scal_{\HBL}$-chain transitive sets, allowing $\HBL$ to be used to characterize $\omega(\Phi)$, and hence $\graph(\Lambda)$.
\end{remark}

\subsection{Reduced System}
In the context of defining the data of the reduced system, let the set-valued mappings $\LRC,\LRD :\Rns \tto \Rnf$ be defined as 
$\graph(\LRC) := \graph(\Lambda)\cap C$ and $\graph(\LRD) := \graph(\Lambda)\cap D$.
Both of these mappings retain the outer-semicontinuity and local boundedness of $\Lambda$ by virtue of their graphs being compact. We claim that the limiting behavior of the slow iterates is characterized by the hybrid inclusion
\begin{equation}\label{eq:reduced_sys}
    \begin{aligned}
        x_s &\in \: \CR \quad \: \dx_s \in \FR(x_s) \\
        x_s &\in \DR \quad \xp_s \in \GR(x_s),
    \end{aligned}
\end{equation}
where $\CR := \dom(\LRC)$, $\DR := \dom(\LRD)$, and the set-valued mappings $\FR,\GR : \Rns \tto \Rns$ are defined as
\[\begin{aligned}
    \FR(x_s) :\!\!&= \cch F_s\!\left(x_s,\LRC(x_s)\right) \\
     &= \cch\left\{f_s \in \Rns : f_s \in F_s(x_s,x_f), \; x_f \in \LRC(x_s)\right\}
\end{aligned}\]
and, with $G_s(x_s,x_f) := \proj_{s}\bigl(G(x_s,x_f)\bigr)$,
\[\begin{aligned}
        \GR(x_s) :\!\!&= G_s(x_s,\LRD(x_s)) \\
        &= \left\{g_s \in \Rns : g_s \in G_s(x_s,x_f), \; x_f \in \LRD(x_s)\right\}.
\end{aligned}\]
We refer to \eqref{eq:reduced_sys} as $\HRL$, or the \emph{reduced system} associated with $\HTT$ and $\Lambda$, with data $\HRL = (\CR,\FR,\DR,\GR)$.
The set-valued mapping $\LRC$, i.e., the restriction of $\Lambda$ to $C$, can be written as $\LRC(x_s) = \{x_f \in \Lambda(x_s) : (x_s,x_f)\in C\}$ for each $x_s \in \CR$, meaning only the flow-admissible values of $x_f = \Lambda(x_s)$ contribute to the behavior of $\FR$. An analogous restriction applies for $\LRD$, $\DR$, and $\GR$.

\begin{lemma}\label{lemma:reduced_regularity}
    Let $(\Phi,\{H_k\}_{k=1}^\infty)$ be a two-timescale asymptotic simulation of $\HTT = (C,F^{1},D,G)$ satisfying Assumption~\ref{ass:basic_conditions}, with $\graph(\Lambda)$ compact and such that $\omega(\Phi) \subset \graph(\Lambda)$. Then, $\HRL = (\CR,\FR,\DR,\GR)$ satisfies Assumption~\ref{ass:basic_conditions}.
\end{lemma}
The following result concerning the behavior of the slow iterates represents the second novel contribution of this paper.

\begin{theorem}\label{theorem:HRL_sim}
    Assume the conditions of Lemma~\ref{lemma:reduced_regularity}. Then, the pair $(\phi_{s},\{h_{s,k}\}_{k=1}^\infty)$ is an asymptotic simulation of $\HRL$.
\end{theorem}

The following corollary provides a further refinement of the limit set of the overall two-timescale asymptotic simulation.

\begin{corollary}\label{corr:final_limit}
    Let the conditions of Theorem~\ref{theorem:HRL_sim} hold. Then,
    \begin{equation*}
        \omega(\Phi) \subset \graph(\Lambda) \cap \left(\Kcal \times \Rnf\right),
    \end{equation*}
    where $\Kcal$ is a compact set containing all internally $\Scal_{\HRL}$-chain transitive and weakly $\Scal_{\HRL}$-invariant sets. Moreover, the set $\omega(\phi_{s})$ is itself nonempty, compact, internally $\Scal_{\HRL}$-chain transitive, and weakly $\Scal_{\HRL}$-invariant.
\end{corollary}

\section{Example - Stochastic Hybrid Optimization}\label{sec:example}
We now demonstrate the utility of the preceding results by establishing the convergence of a two-timescale stochastic approximation of a deterministic hybrid optimization algorithm, making references to \cite[Sec.~6]{GoebelChainTransitivity2026} as needed. Given that we work with a probability space $\Pspace$, we adopt the notation $\olim(\phi) := \omega(\phi)$ to differentiate between $\omega$-limit sets and events $\omega \in \Omega$. The following example is merely intended to illustrate a potential application; it should not be expected to serve as a comprehensive presentation of two-timescale stochastic approximations of hybrid systems.

\subsection{Hybrid Heavy Ball Algorithm}

Consider the optimization problem $\min_{q \in \Rn}\Psi(q)$, where $\Psi : \Rn \to \R$ satisfies the following assumptions.

\begin{assumption}\label{ass:Psi_prop}
    The function $\Psi$ is continuously differentiable, has compact sub-level sets, has a globally Lipschitz gradient $\nabla \Psi$, and, with $\Psi^* := \min_{q \in \Rn}\Psi(q)$, satisfies
    \[\Qcal^* := \{q \in \Rn : \Psi(q) = \Psi^*\} = \{q \in \Rn : \nabla \Psi(q) = 0\}.\]
    Additionally, $\Psi$ is a finite sum with $\Psi(q) := \frac{1}{N}\sum_{i=1}^{N}\Psi_{i}(q)$, $N \in \Z_{\geq 1}$, and $\Psi_i \in C^1$ for each $i \in \{1,\ldots,N\}$.
\end{assumption}

Consider a modified version of the \emph{Hybrid Heavy Ball} (HHB) algorithm of \cite{LeHHB2021}, with state $\chi := (q,p,\tau) \in \R^{2n+1}$,
\begin{equation*}
    \begin{alignedat}{2}
        \chi &\in \CH \quad &&\dchi = \fH(\chi,\nabla\Psi(q)) := \begin{bmatrix}p \\ -\kappa p - \nabla \Psi(q) \\ \min\{1,2-\tfrac{\tau}{T}\}\end{bmatrix}\\
        \chi &\in \DH \quad &&\chip = \gH(\chi) := (q,0,0),
    \end{alignedat}
\end{equation*}
parameters $T > 0$, $\kappa > 0$, and
\[\begin{aligned}
        \CH &:= \left\{\chi : \inner{\nabla\Psi(q),p} \leq 0, \tau \geq T\right\}\cup\{\chi : \tau \leq T\}, \\
        \DH &:= \left\{\chi : \inner{\nabla\Psi(q),p} \geq 0, \tau \geq T\right\}.
\end{aligned}\]
Let $\HHB = (\CH,\fH,\DH,\gH)$ denote this system. The $\tau$ automaton, parametrized by $T > 0$, prevents purely discrete solutions to $\HHB$ by requiring that the system flows when $\tau \in [0,T]$, and is additionally such that the set $[0,2T]$ is globally asymptotically stable (GAS) for $\tau$. The analysis of \cite[Sec.~III]{LeHHB2021} provides a means of selecting an optimal $\kappa$.
By a simple adaptation of the analysis done in \cite{LeHHB2021} for the unmodified HHB system, the compact set $\Mcal := \Qcal^* \times \{0\} \times [0,2T]$ can be shown to be GAS for $\HHB$, meaning all weakly $\Scal_{\HHB}$-invariant and internally $\Scal_{\HHB}$-chain transitive sets are contained in $\Mcal$. Therefore, motivated by Proposition~\ref{prop:aymptotic_sim_omega}, a stochastic approximation scheme that behaves as an asymptotic simulation of $\HHB$ is desirable.

\subsection{Two-Timescale Stochastic Approximation}
An approximation of the aforementioned $\HHB$ system provides a natural setting in which the two-timescale framework yields a concrete benefit over existing single-timescale results. In $\HHB$, the flow set $\CH$ and jump set $\DH$ both depend directly on $\nabla \Psi(q)$, meaning full knowledge of the gradient $\nabla\Psi(q)$ is required to evaluate set membership. A two-timescale approach circumvents this issue by introducing a fast variable that tracks $\nabla \Psi(q)$, and redefines the flow and jump sets to use this fast-timescale estimate, allowing for the slow iterates to asymptotically recover the behavior of $\HHB$.

Let $\Pspace$ be a probability space and $\yp$ be a placeholder for an i.i.d. sequence of random variables $\{\yb_k\}_{k=1}^\infty$, where $\yb_{k} : \Omega \to \R$ satisfies $\yb_k \sim U\{1,\ldots,N\}$ for each $k \in \Z_{\geq 1}$. Take $x_s := \chi = (q,p,\tau)\in\R^{2n+1}$ and introduce a fast variable $x_f := \xi \in \Rn$. 
Consider the following two-timescale simulator, with state $x := (\chi,\xi) \in \R^{3n+1}$,
\begin{equation}\label{eq:TTSA_simulator}
    \begin{alignedat}{2}
        (\chi,\xi) &\in C \quad &&\begin{cases}\chip- \chi=  \hp_s \fH(\chi,\xi) \\ 
        \xip  - \xi = \hp_f(\nabla \Psi_{\yp}(q) - \xi) \end{cases} \\
        (\chi,\xi) &\in D \quad &&(\chip,\xip) = (\gH(\chi), \xi).
    \end{alignedat}
\end{equation}
The parameters $\kappa,T > 0$ are chosen as in $\HHB$, with
\[\begin{aligned}
        C &:= \{(\chi,\xi)  : \inner{\xi,p} \leq 0, \tau \geq T\}\cup\{(\chi,\xi) : \tau \leq T\}, \\
        D &:= \{(\chi,\xi) : \inner{\xi,p} \geq 0, \tau \geq T\}.
\end{aligned}\]
Let $\{\Fcal_{k}\}_{k=0}^\infty$ be the minimal filtration of $\{\yb_{k}\}_{k=1}^\infty$ on $\Fcal$, with $\Fcal_0 := \{\mtset,\Omega\}$. With this, $\Exc{\nabla \Psi_{\yb_{k+1}}(q)}{\Fcal_k} = \nabla\Psi(q)$ for all $k \in \Zgzero$. In the context of discussing solutions to \eqref{eq:TTSA_simulator}, let $\omega \mapsto \Xb(\omega)$ be such that, for every $\omega \in \Omega$, $\Xb(\omega)$ is a hybrid sequence, also called a \emph{sample path}. With $\xb_{r} := \proj_{r}(\Xb)$, the collection of sample paths $\Xb = (\xb_{s},\xb_{f})$ is a \emph{solution} to \eqref{eq:TTSA_simulator} if it is suitably \emph{adapted} (see below) and, for every $\omega \in \Omega$, $\Xb(\omega)$ is a complete hybrid sequence taking values in $\R^{3n+1}$ and satisfying the constraints of \eqref{eq:TTSA_simulator}. The latter condition amounts to requiring that, if $(k,j),(k+1,j) \in \dom\Xb(\omega)$, then $\Xb(k,j) \in C$ and
\begin{equation*}
    \begin{aligned}
        \frac{\xb_{s}(k+1,j)-\xb_{s}(k,j)}{h_{s,k+1}} &= \fH(\xb_{s}(k,j),\xb_{f}(k,j)), \\
        \frac{\xb_{f}(k+1,j)-\xb_{f}(k,j)}{h_{f,k+1}} &= \nabla \Psi_{\yb_{k+1}}(\qb(k,j)) - \xb_{f}(k,j),
    \end{aligned}
\end{equation*}
and, if $(k,j),(k,j+1)\in\dom\Xb(\omega)$, then $\Xb(k,j) \in D$ and
\begin{equation*}
    \begin{aligned}
        \xb_{s}(k,j+1) = \gH(\xb_{s}(k,j)), \quad \xb_{f}(k,j+1) = \xb_{f}(k,j).
    \end{aligned}
\end{equation*}
That $C \cup D = \R^{3n+1}$ ensures maximal solutions to \eqref{eq:TTSA_simulator} are complete, with the $\tau$ automaton additionally guaranteeing completeness in the $k$-direction.
In the stochastic setting the quantities $\jk$ and $\kj$ are both functions of $\omega \in \Omega$, generated from $\dom\Xb(\omega)$ as in \eqref{eq:jk_and_kj}. Adaptedness entails, for each $k \in \Zgzero$, $\Fcal_k$-measurability of the set-valued mapping $\omega \mapsto \graph(\Xb(\omega))\cap(\{k\}\times \R \times\R^{3n+1})$. Adaptedness does not accrue automatically, given that the sets $C$ and $D$ in the simulator \eqref{eq:TTSA_simulator} overlap, allowing for non-unique sample paths and thus collections of sample paths that are not adapted. In short, adaptedness amounts to asking that both the state value $\Xb(\omega)$ and the decision to flow or jump at $(k,j)$ are $\Fcal_{k}$-measurable for each $(k,j) \in \dom\Xb(\omega)$.

\begin{assumption}\label{ass:TTSA_steps}
    The step sizes $\{H_k\}_{k=1}^{\infty} = \{(h_{s,k},h_{f,k})\}_{k=1}^\infty$ are two-timescale admissible and there exists a $p \in [1,\infty)$ for which $\sum_{k=1}^\infty h_{f,k}^{1+p} < \infty$.
\end{assumption}

Under the conditions of Assumptions~\ref{ass:Psi_prop}~and~\ref{ass:TTSA_steps}, the two-timescale system that corresponds to \eqref{eq:TTSA_simulator} is
\begin{equation}\label{eq:TT_HHBsys}
    \begin{alignedat}{2}
        (\chi,\xi) &\in C \quad &&\begin{cases}\dchi = \fH(\chi,\xi) \\ \eps\dxi = -\xi + \nabla \Psi(q) \end{cases} \\
        (\chi,\xi) &\in D \quad &&(\chip,\xip) = (\gH(\chi), \xi).
    \end{alignedat}
\end{equation}
Let $\dx = \FTHBe(x)$ and $\xp = \GTHB(x)$ denote the overall flow and jump dynamics, and $\HTHBe = (C,F_{\mathrm{2HB}}^{\eps},D,G_{\mathrm{2HB}})$ denote \eqref{eq:TT_HHBsys}. The data of $\HTHB$ satisfies Assumption~\ref{ass:basic_conditions}. 

\begin{lemma}\label{lemma:TTSA_sim}
    Let $\Psi$ satisfy Assumption~\ref{ass:Psi_prop}, $\{H_k\}_{k=1}^\infty$ satisfy Assumption~\ref{ass:TTSA_steps}, $\Xb$ be a solution of \eqref{eq:TTSA_simulator}, and $\Omegawhat \subset \Omega$ be such that, for almost every $\omega \in \Omegawhat$, $\Xb(\omega)$ is bounded and complete. Then, for almost every $\omega \in \Omegawhat$, the pair $(\Xb(\omega),\{H_k\}_{k=1}^\infty)$ is a two-timescale asymptotic simulation of $\HTHB$.
\end{lemma}

\begin{theorem}
    Let the conditions of Lemma~\ref{lemma:TTSA_sim} hold. Then, for almost every $\omega \in \Omegawhat$, the pair $(\xb_{s}(\omega),\{h_{s,k}\}_{k=1}^\infty)$ is an asymptotic simulation of $\HHB$, with $\olim(\xb_{s}(\omega)) \subset \Mcal$.
\end{theorem}

\begin{proof}
    By Lemma~\ref{lemma:TTSA_sim}, $(\Xb(\omega),\{H_k\}_{k=1}^{\infty})$ is a two-timescale asymptotic simulation of $\HTHB$ for almost every $\omega \in \Omegawhat$. Applying Theorem~\ref{theorem:HBL_sim}, we can state that, for almost every $\omega \in \Omegawhat$, the pair $(\Xb(\omega),\{h_{f,k}\}_{k=1}^\infty)$ is also an asymptotic simulation of the boundary layer system of $\HTHB$, i.e.,
    \begin{equation}\label{eq:TTSA_BL}
        \begin{alignedat}{2}
            (\chi,\xi) &\in C \quad &&\begin{cases}\dchi = 0 \\ \dxi = -\xi + \nabla \Psi(q) \end{cases} \\
            (\chi,\xi) &\in D \quad &&(\chip,\xip) = (\gH(\chi), \xi).
        \end{alignedat}
    \end{equation}
    Since $\Xb(\omega)$ is bounded, there exists a compact set $K \subset \R^{3n+1}$ for which $\Xb(k,j) \in K$ for all $(k,j) \in \dom\Xb(\omega)$. Take $K_{s} := \proj_{s}(K)$. 
    For any initial condition $(q_0,p_0,\tau_0,\xi_0) \in K$, complete solutions $\psi$ of \eqref{eq:TTSA_BL}, which are complete in the $t$-direction, are such that $q(t,j) \equiv q_0$ for all $(t,j) \in \dom \psi$. Additionally, $\lim_{t + j \to \infty}|\xi(t,j)|_{\nabla \Psi(q_0)} = 0$, since $\xi = \nabla\Psi(q_0)$ is GAS for the $\dxi$ dynamics, solutions are complete in the $t$-direction, and the jump map leaves $\xi$ unchanged. As such, the limit set of \eqref{eq:TTSA_BL} restricted to $K$ is contained in $\graph(\Lambda)$, given as
    \[\graph(\Lambda) = \{(\chi,\xi) 
    \in \R^{3n+1} : \xi = \nabla \Psi(q), \, \chi \in K_{s}\},\] 
    i.e., $\Lambda(\chi) = \{\nabla\Psi(q)\}$ for all $\chi \in \dom(\Lambda) = K_s$. Since $\chi \mapsto \Lambda(\chi)$ is single-valued, we substitute $\xi = \nabla \Psi(q)$ to produce the reduced flow map of $\dchi = \fH(\chi,\nabla\Psi(q))$, with a reduced jump map $\chip = \gH(\chi)$. The mapping $\LRC$ is constructed as $\graph(\LRC) = \graph(\Lambda) \cap C$, meaning $\graph(\LRC) = \{(\chi,\xi)  : \xi = \nabla\Psi(q), \inner{\xi,p} \leq 0, \tau \geq T, \chi \in K_{s}\}\cup\{(\chi,\xi) :\xi = \nabla\Psi(q),  \tau \leq T, \chi \in K_{s}\}$.
    It follows that $\dom(\LRC) = \CH\cap K_s$, with $\dom(\LRD) = \DH\cap K_s$ resulting from an analogous construction. The reduced system, $\HRL$, is therefore
    \begin{equation*}
        \begin{alignedat}{2}
            \chi &\in \CH\cap K_s \quad &&\dchi = \fH(\chi,\nabla\Psi(q)) \\
            \chi &\in \DH\cap K_s \quad &&\chip = \gH(\chi).
        \end{alignedat}
    \end{equation*}
    Since $\CH\cap K_{s} \subset \CH$, $\DH \cap K_{s} \subset \DH$, and the flow and jump maps of $\HRL$ agree with those of $\HHB$, every solution of $\HRL$ is a solution of $\HHB$. Therefore, applying Theorem~\ref{theorem:HRL_sim} and using that $\Scal_{\HRL} \subset \Scal_{\HHB}$, the pair $(\xb_{s}(\omega),\{h_{s,k}\}_{k=1}^\infty)$ is an asymptotic simulation of both $\HRL$ and $\HHB$.
    Furthermore, applying Proposition~\ref{prop:aymptotic_sim_omega}, since $\Mcal$ is GAS for $\HHB$ and therefore contains all internally $\Scal_{\HHB}$-chain transitive and weakly $\Scal_{\HHB}$-invariant sets, it follows that $\olim(\xb_{s}(\omega)) \subset \Mcal$. Since $\omega \in \Omegawhat$ was chosen arbitrarily, possibly excluding a set of measure zero, the claim holds.
\end{proof}

\bibliographystyle{IEEEtran}
\bibliography{./TTSA_references.bib}

\newpage
\begin{appendix}
\renewcommand{\thelemma}{\Alph{section}.\arabic{lemma}}
\setcounter{lemma}{0}

\begin{lemma}\label{lemma:supremums}
    Let the pair of step sizes $\{(h_{s,k},h_{f,k})\}_{k=1}^\infty$ be two-timescale admissible. Then, for each $T > 0$, there exists an $\ell \in \Zgzero$ such that $\Ical_{f,n,T}$ and $\Ical_{s,n,T}$ are nonempty and $\Ical_{f,n,T} \subset \Ical_{s,n,T}$ for all $n \in \Z_{\geq \ell}$.
\end{lemma}

\begin{proof}
    Fix $T > 0$. To have $\Ical_{f,n,T} \subset \Ical_{s,n,T}$ we require that $m_f(\tau_{f,n} + T) \leq m_s(\tau_{s,n} + T)$. From the definitions of $\tau_{r,n}$ and $m_{r}(\cdot)$, we have $m_{r}(\tau_{r,n}+T) \geq n$. As such,
    \[m_r(\tau_{r,n} + T) = \max\left\{k \geq n : \sum_{i=n}^{k-1} h_{r,i+1}\leq T\right\}.\]
    The two-timescale admissibility of the step sizes, specifically condition \eqref{eq:timescale_sep}, guarantees that there exists a $\nu \in \Zgzero$ such that $h_{s,k} \leq h_{f,k}$ for all $k \in \Z_{\geq \nu}$. Therefore, for all $n \in \Z_{\geq \nu}$,
    \[\max\left\{k  : \sum_{i=n}^{k-1} h_{f,i+1}\leq T\right\} \leq \max\left\{k  : \sum_{i=n}^{k-1} h_{s,i+1}\leq T\right\},\]
    i.e., $m_f(\tau_{f,n} + T) \leq m_s(\tau_{s,n} + T)$ for all $n \geq \nu$. Nonemptiness of $\Ical_{r,n,T}$ occurs when $h_{r,n+1} \leq T$. The admissibility of both sequences of step sizes guarantees the existence of an $N \in \Zgzero$ for which $h_{r,n+1} \leq T$ for all $n \geq N$ and $r \in \{s,f\}$. Therefore, both index sets are nonempty for all $n \in \Z_{\geq N}$. Taking $\ell := \max\{\nu,N\}$, establishes nonemptiness of both sets and that $\Ical_{f,n,T} \subset \Ical_{s,n,T}$ for all $n \geq \ell$. Given that our choice of $T$ was arbitrary, the claim holds.
\end{proof}

For the remainder of the work we will, with some abuse of notation, write
\[(\Phi(k,j),f_{k}) = \bigl(\phi_{s}(k,j),\phi_{f}(k,j),f_{s,k},f_{f,k}\bigr).\] 
Additionally, define for each $(r,n) \in \{s,f\}\times\Zgzero$, $k \in \Z_{\geq n+1}$, hybrid sequence $\phi_{q}$, and bounded sequence of vectors $\{f_{p,k}\}_{k=1}^\infty$ the quantity
\[\Ecal_{r,n,\{f_{p,k}\}_{k=1}^\infty}^{k,\phi_{q}} := \left|\sum_{i=n}^{k-1}h_{r,i+1}\left(\tfrac{\phi_{q}(i+1,\jmol_{i})-\phi_{q}(i,\jmol_{i})}{h_{r,i+1}}-f_{p,i}\right)\right|.\]
Note that the $r$-parameter encodes the relevant timescale, i.e., choice of step sizes $\{h_{r,k}\}_{k=1}^\infty$.

\begin{proof}[Proof of Theorem~\ref{theorem:HBL_sim}]
    Let $(\Phi,\{H_k\}_{k=1}^\infty)$ be a two-timescale asymptotic simulation of $\HTT$. For $(\Phi,\{h_{f,k}\}_{k=1}^\infty)$ to be an asymptotic simulation of $\HBL$ we require that, if $\Phi$ is complete in the $k$-direction, there exists a bounded sequence of vectors $\{\fwtild_{k}\}_{k=0}^\infty$ such that
    \begin{equation}\label{eq:fasttimescale_subset}
        \limsup_{k \to \infty}\left(\Phi(k,\jk),\fwtild_{k}\right) \subset \graph(\Fwtild_C),
        \vspace{-4pt}
    \end{equation}
    and, for each $T > 0$, $\displaystyle \lim_{n \to \infty}\sup_{k \in \Ical_{f,n,T}}\Ecal_{f,n,\{\fwtild_{k}\}_{k=1}^\infty}^{k,\Phi} = 0$, i.e.,
    \vspace{-5pt}
    \begin{equation}\label{eq:two_timscale_limsup}
        \lim_{n \to \infty}\sup_{k \in \Ical_{f,n,T}}\left|\sum_{i=n}^{k-1}h_{f,i+1}\left(\tfrac{\Phi(i+1,\ji)-\Phi(i,\ji)}{h_{f,i+1}}-\fwtild_{i}\right)\right| = 0.
    \end{equation}
    Take $\{f_{k}\}_{k=0}^\infty$ to be the bounded sequence of vectors from Definition~\ref{def:two_timescale_sim} associated with $(\Phi,\{H_k\}_{k=1}^\infty)$, and define
    \[\fwtild_{s,k} :=
        \frac{h_{s,k+1}}{h_{f,k+1}}f_{s,k}, \quad \fwtild_{f,k} := f_{f,k}, \quad \fwtild_{k} := \begin{bmatrix}
        \fwtild_{s,k} \\ \fwtild_{f,k}
    \end{bmatrix}.\]
    The boundedness of $\{f_{k}\}_{k=1}^\infty$ and $\lim_{k\to\infty}h_{s,k}/h_{f,k} = 0$ ensures that $\{\fwtild_{k}\}_{k=1}^\infty$ is bounded. Fix an accumulation point
    \[(\xi_s,\xi_f,f_s^\star,f_f^\star) \in \limsup_{k \to \infty}(\Phi(k,\jk),f_{k}),\] i.e., there exists a subsequence $\{k_n\}_{n=1}^\infty$ along which $\lim_{n \to \infty}(\Phi(k_n,\jmol_{k_n}),f_{k_n}) = (\xi_s,\xi_f,f_s^\star,f_f^\star)$. Using the boundedness of $\{f_{s,k}\}_{k=1}^\infty$ and \eqref{eq:timescale_sep}, we have that
    \[\begin{aligned}
        \lim_{n \to \infty}\fwtild_{k_n} 
        = \lim_{n \to \infty}\left(\frac{h_{s,k_n+1}}{h_{f,k_n+1}}f_{s,k_n},f_{f,k_n}\right) = (0,f_f^\star)
    \end{aligned},\]
    implying $\lim_{n \to \infty}(\Phi(k_n,\jmol_{k_n}),\fwtild_{k_n}) = (\xi_s,\xi_f,0,f_f^\star)$.
    From \eqref{eq:two_timescale_F_containment}, we know $(\xi_s,\xi_f,f_s^\star,f_f^\star) \in \graph(F^{1}_C)$, and, using the definition of $\Fwtild$, we can state that $(\xi_s,\xi_f,0,f_f^\star) \in \graph(\Fwtild_C)$
    for all convergent subsequences $\{k_n\}_{n=1}^\infty$, meaning \eqref{eq:fasttimescale_subset} is satisfied.
    Note that \eqref{eq:two_timscale_limsup} is equivalent to having 
    \begin{subequations}
        \begin{equation}\label{eq:slow_closeness_fast_timescale}
            \lim_{n \to \infty}\sup_{k \in \Ical_{f,n,T}}\Ecal_{f,n,\{\fwtild_{s,k}\}_{k=1}^\infty}^{k,\phi_{s}} = 0
        \end{equation}
        \begin{equation}\label{eq:fast_closeness_fast_timescale}
            \lim_{n \to \infty}\sup_{k \in \Ical_{f,n,T}}\Ecal_{f,n,\{\fwtild_{f,k}\}_{k=1}^\infty}^{k,\phi_{f}} = 0
        \end{equation}
    \end{subequations}
    hold for all $T > 0$. That \eqref{eq:fast_closeness_fast_timescale} holds is immediate given the assumption of \eqref{eq:fast_closeness}. It must be shown that an assumption of \eqref{eq:slow_closeness} implies \eqref{eq:slow_closeness_fast_timescale}. The argument of the supremum in \eqref{eq:slow_closeness} may be written as
    \[\begin{aligned}
        \left|\sum_{i=n}^{k-1}h_{f,i+1}\left(\frac{\phi_{s}(i+1,\ji)-\phi_{s}(i,\ji)}{h_{f,i+1}} - \frac{h_{s,i+1}}{h_{f,i+1}}f_{s,i}\right)\right|,
    \end{aligned}\]
    making \eqref{eq:slow_closeness} of Definition~\ref{def:two_timescale_sim} equivalent to
    \begin{equation}\label{eq:slow_closeness_fast_slowclock}
        \lim_{n \to \infty}\sup_{k \in \Ical_{s,n,T}}\Ecal_{f,n,\{\fwtild_{s,k}\}_{k=1}^\infty}^{k,\phi_{s}} = 0.
    \end{equation}
    The condition \eqref{eq:slow_closeness_fast_timescale} can therefore be claimed by showing that \eqref{eq:slow_closeness_fast_slowclock} holds with $\Ical_{f,n,T}$ in place of $\Ical_{s,n,T}$. Note that, for all $n \in \Zgzero$, $T > 0$, and $r \in \{s,f\}$, the quantity $\Ecal_{f,n,\{\fwtild_{s,k}\}_{k=1}^\infty}^{k,\phi_{s}}$ is non-negative for each $k \in \Ical_{r,n,T}$. Applying Lemma~\ref{lemma:supremums} and the monotonicity of the supremum, for each $T > 0$ there exists an $\ell \in \Zgzero$ such that, for all $n \in \Z_{\geq \ell}$,
    \[0 \leq \sup_{k \in \Ical_{f,n,T}}\Ecal_{f,n,\{\fwtild_{s,k}\}_{k=1}^\infty}^{k,\phi_{s}} \leq \sup_{k \in \Ical_{s,n,T}}\Ecal_{f,n,\{\fwtild_{s,k}\}_{k=1}^\infty}^{k,\phi_{s}}.\]
   Using \eqref{eq:slow_closeness_fast_slowclock}, we can therefore conclude
    \begin{equation}\label{eq:bl_fast_limit}
        \lim_{n \to \infty}\sup_{k \in \Ical_{f,n,T}}\Ecal_{f,n,\{\fwtild_{s,k}\}_{k=1}^\infty}^{k,\phi_{s}} = 0.
    \end{equation}
    Since \eqref{eq:slow_closeness} was assumed for all $T > 0$, \eqref{eq:bl_fast_limit} holds for all $T > 0$, establishing \eqref{eq:slow_closeness_fast_timescale}. Finally, if $\Phi$ is complete in the $j$-direction, the property \eqref{eq:two_timescale_G_containment} holds automatically. Consequently, the pair $(\Phi,\{h_{f,k}\}_{k=1}^\infty)$ satisfies \eqref{eq:asymptotic_sim_F_containment}, \eqref{eq:asymptotic_sim_closeness}, and \eqref{eq:asymptotic_sim_G_containment} for $\HBL$, and is therefore an asymptotic simulation of $\HBL$.
\end{proof}

\begin{proof}[Proof of Lemma~\ref{lemma:HBL_omegalim}]
    Since $\Phi(k,j) \in K$ for all $(k,j) \in \dom\Phi$, the asymptotic simulation of $\HBL$ generated by $(\Phi,\{h_{f,k}\}_{k=1}^\infty)$ must also be one of $\HBLK$. By \cite[Prop.~6.28]{HDS_Book}, the set $\OmegaBLK(K)$ is the minimal global attractor for $\HBLK$. Consequently, all weakly $\Scal_{\HBLK}$-invariant and internally $\Scal_{\HBLK}$-chain transitive sets must be contained within. By Proposition~\ref{prop:aymptotic_sim_omega}, $\omega(\Phi)$ is such a set, meaning $\omega(\Phi) \subset \OmegaBLK(K)$.
\end{proof}

\begin{proof}[Proof of Lemma~\ref{lemma:reduced_regularity}]
    The sets $\CR$ and $\DR$ are compact (and hence closed) by the compactness of $\graph(\Lambda)$ and the closedness of $C$ and $D$.
    That $\CR = \dom(\LRC)$, $\DR = \dom(\LRD)$, $F_s$ is nonempty on $C$, and $G_s$ is nonempty on $D$ ensures that both $F_s(x_s,\LRC(x_s))$ and $G_s(x_s,\LRD(x_s))$ are nonempty for all $x_s \in \CR$ and $x_s \in \DR$, respectively.
    An application of \cite[Prop.~5.52(a),(b)]{VariationalAnalysis} guarantees that the respective compositions of the set-valued mappings $F_s$ and $G_s$ with $\LRC$ and $\LRD$ preserve outer semicontinuity and local boundedness, meaning both $\GR : \Rns \tto \Rns$ and $x_s \mapsto F_s(x_s,\LRC(x_s))$ are outer-semicontinuous and locally bounded.
    By \cite[Corollary~2.30]{VariationalAnalysis}, the closed convex hull preserves the closedness of the graph, and therefore the compactness of its range, meaning $x_s \mapsto \cch F_s(x_s,\LRC(x_s))$, i.e., $x_s \mapsto \FR(x_s)$, is locally bounded, outer-semicontinuous, and is nonempty and has convex values for each $x_s \in \CR$. As such, the data of $\HRL$ satisfies Assumption~\ref{ass:basic_conditions}. 
\end{proof}

\begin{proof}[Proof of Theorem~\ref{theorem:HRL_sim}]
    Consider the two-timescale asymptotic simulation $(\Phi,\{H_k\}_{k=1}^\infty)$ of $\HTT$. For $(\phi_{s},\{h_{s,k}\}_{k=1}^\infty)$ to constitute an asymptotic simulation of $\HRL$, there must exist a bounded sequence of vectors $\{\fbar_{s,k}\}_{k=1}^\infty$ for which 
    \begin{equation}\label{eq:slowtimescale_subset}
        \limsup_{k \to \infty}\left(\phi_{s}(k,\jk),\fbar_{s,k}\right) \subset \graph(\FR)\cap(\CR\times \Rns)
    \end{equation}
    and 
    \begin{equation}\label{eq:reduced_sys_closeness}
        \lim_{n \to \infty}\sup_{k \in \Ical_{s,n,T}}\left|\sum_{i=n}^{k-1}h_{s,i+1}(\fwhat_{s,i+1}-\fbar_{s,i})\right| = 0,
    \end{equation}
    for all $T > 0$.
    Take $\{\fbar_{s,k}\}_{k=1}^\infty = \{f_{s,k}\}_{k=1}^\infty$, which ensures that conditions \eqref{eq:two_timescale_F_containment} and \eqref{eq:slow_closeness} hold. It is immediate that with this choice of $\{\fbar_{s,k}\}_{k=1}^\infty$, \eqref{eq:slow_closeness} implies \eqref{eq:reduced_sys_closeness}. For \eqref{eq:slowtimescale_subset} to hold, it must be shown that all accumulation points of $\{(\phi_{s}(k,\jk),f_{s,k})\}_{k=1}^\infty$ are contained in $\graph(\FR)\cap(\CR\times\Rns)$. Fix a point $(\xi_s,f_s^\star) \in \limsup_{k \to \infty}(\phi_{s}(k,\jk),f_{s,k})$ and its corresponding subsequence $\{k_n\}_{n=1}^\infty$. By virtue of $\{\Phi(k,\jk)\}_{k=1}^\infty$ and $\{f_{k}\}_{k=1}^\infty$ being bounded, pass, without relabeling, to a further subsequence of $\{k_n\}_{n=1}^\infty$ satisfying $\lim_{n \to \infty}(\Phi(k_n,\jmol_{k_n}),f_{k_n}) = (\xi_s,\xi_f,f_s^\star,f_f^\star)$.
    By the definition of the outer limit and \eqref{eq:two_timescale_F_containment},
    \begin{equation}\label{eq:subseq_limits}
        (\xi_s,\xi_f,f_s^\star,f_f^\star) \in  \graph(F^{1}_C),
    \end{equation}
    implying $(\xi_s,\xi_f) \in C$ and $(f_s^\star,f_f^\star) \in F^{1}(\xi_s,\xi_f)$. Combining \eqref{eq:limsup_k_A} and \eqref{eq:subseq_limits} yields $(\xi_s,\xi_f) \in \graph(\Lambda)\cap C$, and consequently $\xi_s \in \CR$ and $\xi_f \in \LRC(\xi_s)$. As such,
    \[f_s^\star \in F_s(\xi_s,\xi_f) \subset F_s(\xi_s,\LRC(\xi_s))  \subset \FR(\xi_s).\]
    Since the above conclusions apply for all accumulation points $(\xi_s,\xi_f,f_s^\star,f_f^\star)$, the sequence $\{(\phi_{s}(k,\jk),f_{s,k})\}_{k=1}^\infty$ satisfies \eqref{eq:slowtimescale_subset}. If $\Phi$ is complete in the $j$-direction, we can combine \eqref{eq:two_timescale_G_containment} and \eqref{eq:limsup_j_A}, and use a similar argument to get $\limsup_{j \to \infty}(\phi_{s}(\kj,j),\phi_{s}(\kj,j+1)) \subset \graph(\GR)\cap(\DR \times \Rns)$. With the necessary conditions met, we may state that $(\phi_{s},\{h_{s,k}\}_{k=1}^\infty)$ is an asymptotic simulation of $\HRL$.
\end{proof}

\begin{proof}[Proof of Corollary~\ref{corr:final_limit}]
    Proposition~\ref{prop:aymptotic_sim_omega} and Theorem~\ref{theorem:HRL_sim} allow us to claim that $\omega(\phi_{s})$ is a nonempty, compact set that is internally $\Scal_{\HRL}$-chain transitive and weakly $\Scal_{\HRL}$-invariant. Let $\Kcal$ be a compact set containing every such set. Fix $(\xi_s,\xi_f) \in \omega(\Phi)$. Since $\omega(\phi_{s}) = \proj_{s}(\omega(\Phi))$, we have $\xi_s \in \omega(\phi_{s})\subset \Kcal$. Furthermore, since $(\xi_s,\xi_f)\in\graph(\Lambda)$, we can state that $(\xi_s,\xi_f)\in\graph(\Lambda)\cap(\Kcal\times\Rnf)$. Given $(\xi_s,\xi_f)$ was selected arbitrarily, the claim holds.
\end{proof}

\begin{proof}[Proof of Lemma~\ref{lemma:TTSA_sim}]
    Adapting \cite[Sec.~6.3]{GoebelChainTransitivity2026}, define, for all $(k,\jk),(k+1,\jk) \in \dom\Xb(\omega)$, the quantity
    \[\fbwhat_{r,k+1} := \frac{\xb_{r}(k+1,\jk)-\xb_{r}(k,\jk)}{h_{r,k+1}}.\]
    The data of \eqref{eq:TTSA_simulator} is such that every maximal solution is complete in the $k$-direction, meaning $\fbwhat_{r,k+1}$ is well defined for almost all $\omega \in \Omega$ and $k \in \Zgzero$. Let 
    \[\fb_{r,k} := \Exc{\fbwhat_{r,k+1}}{\Fcal_{k}}, \quad \vb_{r,k+1} := \fbwhat_{r,k+1} - \fb_{r,k},\]
    and $\Fb_{k} := (\fb_{s,k},\fb_{f,k})$.
    For a bounded $\Xb(\omega)$ to be a two-timescale asymptotic simulation of $\HTHB$ we require that
    \begin{equation}\label{eq:TTSA_limsup}
        \limsup_{k \to \infty}(\Xb(k,\jk),\Fb_k) \subset \graph((\FTHB)_C).
    \end{equation}
    This is a stochastic analogue of \eqref{eq:two_timescale_F_containment}.
    Additionally, \cite[Asm.~6.7]{GoebelChainTransitivity2026} must be met, i.e., there must exist, for $r \in \{s,f\}$, both a $p \in [1,\infty)$ satisfying $\sum_{k=1}^\infty h_{r,k}^{1+p} < \infty$ and a continuous, nondecreasing function $\gamma : \Rgzero \to \Rgzero$ where
    \begin{equation}\label{eq:TTSA_variance}
        \Exc{|\vb_{r,k+1}|^{2p}}{\Fcal_k} \leq \gamma(|\Xb(k,\jk)|) \quad \forall k \in \Zgzero.
    \end{equation}
    By \cite[Prop.~6.8]{GoebelChainTransitivity2026} the dual requirement on \eqref{eq:TTSA_variance} and $\{h_{r,k}\}_{k=1}^\infty$ endows sample paths $\Xb(\omega)$ with an equivalent of condition \eqref{eq:two_timescale_closeness}. These conditions must be shown to hold for almost every $\omega \in \Omegawhat$. 
    Given that the slow dynamics are deterministic, $\fbwhat_{s,k+1} = \fb_{s,k}$, meaning $\vb_{s,k+1} = 0$ and the $r=s$ conditions are vacuous. We may therefore restrict our attention to $\fb_{f,k}$, $\{h_{f,k}\}_{k=1}^\infty$, and $\vb_{f,k+1}$. Since $\Exc{\nabla \Psi_{\yb_{k+1}}(q)}{\Fcal_k} = \nabla\Psi(q)$, we have $\fb_{f,k} = \nabla \Psi(\qb(k,\jk)) - \xb_{f}(k,\jk)$, meaning
    \[\Fb_k = \begin{bmatrix}\fb_{s,k} \\ \fb_{f,k} \end{bmatrix} = \begin{bmatrix}\fH(\xb_{s}(k,\jk),\xb_{f}(k,\jk)) \\ \nabla \Psi(\qb(k,\jk)) - \xb_{f}(k,\jk) \end{bmatrix}.\]
    As such, whenever $\Xb(k,\jk) \in C$, we have $(\Xb(k,\jk),\Fb_k) \in \graph((\FTHB)_C)$, satisfying \eqref{eq:TTSA_limsup}. Let $p \in [1,\infty)$ be as given by Assumption~\ref{ass:TTSA_steps}. Since for all $k \in \Zgzero$ it holds that
    \[|\vb_{f,k+1}|^{2p} \leq \max_{i \in \{1,\ldots,N\}}|\nabla \Psi_{i}(\qb(k,\jk)) - \nabla \Psi(\qb(k,\jk))|^{2p},\]
    we may write $\Exc{|\vb_{f,k+1}|^{2p}}{\Fcal_k} \leq B_{p}(\qb(k,\jk))$, where
    \[B_{p}(q) :=  \max_{i \in \{1,\ldots,N\}}|\nabla \Psi_{i}(q) - \nabla \Psi(q) |^{2p}.\]
    The continuity of $q \mapsto B_{p}(q)$ follows from $\Psi,\Psi_i \in C^1$. Taking $\gamma(r) := \max_{|q|\leq r}B_{p}(q)$, we have that $\gamma$ is continuous and nondecreasing. Therefore, for all $k \in \Zgzero$,
    \[\Exc{|\vb_{f,k+1}|^{2p}}{\Fcal_k} \leq B_{p}(\qb(k,\jk))\leq \gamma(|\Xb(k,\jk)|),\]
    since $|\qb(k,\jk)|\leq|\Xb(k,\jk)|$. As such, \eqref{eq:TTSA_variance} is satisfied.
    Given that the jump maps of \eqref{eq:TTSA_simulator} and $\HTHB$ are identical, the asymptotic simulation conditions are trivially satisfied for the jump dynamics. Therefore, for almost every $\omega \in \Omegawhat$, the pair $(\Xb(\omega),\{H_k\}_{k=1}^\infty)$ is a two-timescale asymptotic simulation of $\HTHB$.
\end{proof}

\end{appendix}

\end{document}